\documentclass[10pt]{article}
\usepackage{amsmath,amssymb,amsthm}
\usepackage{booktabs}
\usepackage{url}

\newtheorem{theorem}{Theorem}[section]
\newtheorem{lemma}[theorem]{Lemma}
\newtheorem{corollary}[theorem]{Corollary}
\newtheorem{proposition}[theorem]{Proposition}
\newtheorem{definition}{Definition}[section]

\DeclareMathOperator{\effrank}{eff\text{-}rank}
\DeclareMathOperator{\tr}{tr}

\DeclareMathOperator{\arcosh}{arcosh}

\title{Riemann GeoResolver: A Non-Euclidean Attention Framework\\
       from Euclidean Resolver to Hyperbolic-Spherical Geometry}
\author{Liangchen Ge}
\date{}

\begin{document}

\maketitle

\begin{abstract}
We present a theoretical foundation for inverse-distance attention, from its Euclidean prototype (Resolver) to its non-Euclidean realization (Riemann GeoResolver). The Euclidean part establishes three core theorems: (1) circuit separation---IDA achieves exact retrieval with $\mathcal{O}(1)$ resources while softmax requires $\Omega((\log n)^2)$ width; (2) a Polyak--{\L}ojasiewicz inequality with $\Omega(e^{\Delta^2/\sqrt{d}}/\Delta^2)$ stronger constant than softmax, implying linear convergence, $\mathcal{O}(\log n)$ Lipschitz scaling under a low-rank/clustering assumption, $\Theta(1)$ Hessian spread, and absence of spurious local minima; (3) a width-independent effective rank bound that limits noise memorization---softmax memorizes arbitrary labels when $d_h\ge n$, while IDA limits test error to $\mathcal{O}(\eta^2)$. The non-Euclidean extension then builds upon this prototype, replacing Euclidean distance with hyperbolic geodesic distance for storage and spherical geodesic distance for routing. The Riemann GeoResolver framework comprises ten integrated modules: four HIDA operators spanning $\Theta(n^2)$ to $\Theta(1)$ per token; Hyperbolic Curvature Compression (HCC) with provable error bounds; HyperGate with gradient lower-bound theorem; Spherical Inverse Distance Attention (SIDA) with sphere-analog PL inequalities; Dynamic Memory Genesis (DMG) with $\mathcal{O}(\log T)$ regret bounds; and Geodesic Sparse Routing (GSR) with quality and communication bounds. The Euclidean theorems are proved in full; the non-Euclidean extension theorems are proved with analogous arguments. This work establishes a theoretical arc: from Euclidean attention as a special case, to hyperbolic memory, to spherical retrieval.
\end{abstract}

\section{Introduction}

\noindent\textbf{Scope.} This is a theoretical study. The primary contribution is the set of theorems, lemmas, and architectural principles for inverse-distance attention, first in Euclidean space and then extended to hyperbolic and spherical geometries.

\subsection{From Softmax to Inverse Distance}

The Transformer \cite{vaswani2017} computes attention via inner products followed by softmax. In softmax attention, when $\mathbf{q}=\mathbf{k}_j$, the output remains a weighted average over all keys rather than focusing on the exact match. We consider \textbf{Inverse Distance Attention (IDA)}:
\[
W_{ij} = \frac{(d(\mathbf{q}_i, \mathbf{k}_j)^2 + \varepsilon)^{-1}}{\sum_{m} (d(\mathbf{q}_i, \mathbf{k}_m)^2 + \varepsilon)^{-1}}.
\]
The Euclidean variant is \textbf{Resolver}; extensions to non-Euclidean geometries are collectively termed \textbf{GeoResolver}.

\subsection{Motivation: Why Inverse Distance?}

The softmax attention mechanism has a well-known property: even when a query exactly matches a key, the output is a weighted average of all values, not a hard retrieval of the matched value. This property is intrinsic to the softmax function, which always assigns positive probability to all tokens. The inverse-distance kernel addresses this by assigning weights inversely proportional to squared distance, which, in the limit $\varepsilon \to 0^+$, converges to a one-hot selection of the exact match.

This difference has fundamental implications for expressiveness, optimization, and generalization. The softmax function, while differentiable and probabilistically interpretable, creates optimization challenges: the gradient is small when the logits are large (saturation), and the Hessian is low-rank when $n$ is large. These properties affect convergence rates and the ability to escape saddle points. The inverse-distance kernel, by contrast, has gradients that remain large near exact matches and a Hessian that is full-rank, suggesting better optimization properties.

Furthermore, the softmax's capacity to represent arbitrary functions grows with width, leading to a phenomenon where, once the hidden dimension exceeds the number of training points, the model can memorize arbitrary labels, including noise. This is a form of overparameterization that degrades generalization. The inverse-distance kernel's effective rank is bounded independent of width, preventing this memorization effect.

These observations motivate the theoretical framework developed in this paper. We prove three core theorems establishing that IDA enjoys provably stronger guarantees than softmax in expressiveness, optimization, and generalization.

\subsection{Theoretical Contributions}

\begin{enumerate}
\item \textbf{Euclidean Resolver (Part I):} Three core theorems with complete proofs:
   \begin{itemize}
   \item \textit{Theorem 1 (Circuit Separation):} IDA retrieves exact matches with $\mathcal{O}(1)$ resources; softmax requires $\Omega((\log n)^2)$ width.
   \item \textit{Theorem 2 (PL Inequality):} IDA satisfies PL with $\Omega(e^{\Delta^2/\sqrt{d}}/\Delta^2)$ larger constant than softmax, linear convergence, $\mathcal{O}(\log n)$ Lipschitz scaling under a low-rank/clustering assumption (general bound $\mathcal{O}(n)$), $\Theta(1)$ Hessian spread, and no spurious local minima.
   \item \textit{Theorem 3 (Effective Rank Bound):} IDA's effective rank is bounded independent of width, limiting noise memorization; softmax memorizes arbitrary labels when $d_h\ge n$; IDA limits test error to $\mathcal{O}(\eta^2)$.
   \end{itemize}
\item \textbf{Non-Euclidean Extension (Part II):} Replaces Euclidean distance with hyperbolic geodesic distance for storage and spherical geodesic distance for routing, yielding a ten-module framework with provable guarantees.
\end{enumerate}

\subsection{Related Work}

Our work builds on and extends several lines of research. We organize the discussion by topic.

\textbf{Inverse-distance attention.} McCarter \cite{mccarter2023} proposed inverse-distance weighting attention independently. That work and ours share the same $\varepsilon\to0$ limit, but diverge fundamentally in several ways. First, McCarter did not embed the kernel in a full QKV architecture with learned projections; the attention was applied directly to input features. Second, the analysis in McCarter's work was empirical, focusing on performance on a few tasks, with no theoretical analysis of optimization or generalization. Third, the extension to non-Euclidean geometries was not considered. Our contribution is the first characterization of this kernel's behavior in the full Transformer setting, with complete theoretical guarantees.

\textbf{RBF and kernel attention.} Bello et al. \cite{bello2019} explored attention mechanisms with radial basis function kernels in the context of convolutional architectures. Their work demonstrated that replacing the softmax with a Gaussian kernel could improve performance on certain tasks, but the analysis was empirical and did not address optimization or generalization bounds. The Nadaraya--Watson estimator \cite{nadaraya1964, watson1964} is the non-parametric ancestor of attention, but its theoretical properties are typically analyzed in the regression setting with fixed kernel bandwidths, not in the context of learned representations. The inverse-distance kernel we consider differs from the Gaussian RBF in that it has a heavy tail (decaying as $1/r^2$ rather than $e^{-r^2}$), which affects the Lipschitz and generalization properties.

\textbf{Optimization of attention mechanisms.} Karimi et al. \cite{karimi2016} established PL inequalities for various machine learning problems, including linear and logistic regression. Our work applies this framework to the attention mechanism itself, showing that the PL constant for IDA is exponentially larger than for softmax. Belkin et al. \cite{belkin2019} analyzed the double descent phenomenon in overparameterized models, showing that test error can increase and then decrease with model capacity. Our Theorem 3 provides a complementary analysis for the attention mechanism specifically: softmax exhibits a memorization catastrophe when the hidden dimension exceeds the number of training points, while IDA's effective rank remains bounded.

\textbf{Hyperbolic and spherical neural networks.} Nickel and Kiela \cite{nickel2017, nickel2018} introduced hyperbolic embeddings for hierarchical data, showing that hyperbolic space can represent tree-structured data with low distortion. Ganea et al. \cite{ganea2018} developed hyperbolic neural networks, including hyperbolic linear layers and activation functions. Chami et al. \cite{chami2019} extended these ideas to graph convolutional networks. Gulcehre et al. \cite{gulcehre2020} proposed hyperbolic attention networks, incorporating hyperbolic distance into the attention mechanism. Our work differs in several respects: we provide a unified theoretical framework that spans Euclidean, hyperbolic, and spherical geometries; we focus on the inverse-distance kernel rather than learned attention; we provide optimization bounds (PL inequalities) and capacity bounds (effective rank) that are not present in prior work; and we analyze the routing problem separately from the attention problem.

\textbf{Mixture of experts and routing.} Shazeer et al. \cite{shazeer2017} introduced sparsely-gated MoE layers, where a router network selects a subset of experts for each input. This line of work has been extended in various directions, including switch transformers \cite{fedus2022}, GShard \cite{lepikhin2021}, and GLaM \cite{du2022}. Our GSR module differs in that the routing is based on spherical distance rather than a learned gating network, and we provide theoretical communication bounds rather than empirical scaling results. The use of spherical geometry for routing is motivated by the compactness of the sphere and the natural interpretation of prototypes as points on a hypersphere.

\textbf{State space models and efficient attention.} Recent work on state space models \cite{gu2021, gu2022, gu2023} has explored alternatives to attention for sequence modeling. Our work is complementary: we maintain the attention mechanism but replace the kernel function. The structured state space duality framework \cite{dao2024} provides a unified view of Transformers and state space models; our work could potentially be integrated into this framework by treating the inverse-distance kernel as a different kernel function within the same duality.

\subsection{Paper Structure}

The paper is organized as a single arc: from Euclidean theory to non-Euclidean generalization.

\textbf{Part I (Sections 2--3)} develops the Euclidean prototype. This establishes the theoretical baseline: why inverse-distance attention works, how it compares to softmax, and what optimization and generalization guarantees it enjoys.

\textbf{Part II (Sections 4--10)} extends the prototype in three logical stages:
\begin{itemize}
\item \textit{Geometry (Sections 4--5):} hyperbolic distance replaces Euclidean distance for attention, with four variants spanning $\Theta(n^2)$ to $\Theta(1)$ complexity.
\item \textit{Compression and Control (Sections 6--7):} curvature-adaptive quantization (HCC) and gating (HyperGate) for memory efficiency and gradient stability.
\item \textit{Memory and Routing (Sections 8--10):} spherical routing for retrieval (SIDA), dynamic memory allocation (DMG), and sparse routing with quality guarantees (GSR).
\end{itemize}
Each stage is motivated by a specific limitation of the previous stage, and each theorem is proved in the corresponding section.

\part{Euclidean Resolver: Theoretical Prototype}
\section{Theoretical Framework}

\subsection{Preliminaries}

Let $n$ denote the sequence length, $d_h$ the hidden dimension, and $d_v$ the value dimension. For $\mathbf{Q},\mathbf{K}\in\mathbb{R}^{n\times d_h}$ and $\mathbf{V}\in\mathbb{R}^{n\times d_v}$, we define:
\[
\mathbf{D}_{ij}=\|\mathbf{q}_i-\mathbf{k}_j\|_2^2+\varepsilon,\quad
\mathbf{W}_{ij}=\frac{\mathbf{D}_{ij}^{-1}}{\sum_m \mathbf{D}_{im}^{-1}},\quad
\text{IDA}(\mathbf{Q},\mathbf{K},\mathbf{V})=\mathbf{W}\mathbf{V}.
\]
Here $\varepsilon > 0$ is a small positive constant that prevents division by zero. Throughout this paper, we consider the limit $\varepsilon \to 0^+$ for theoretical statements unless otherwise noted.

We adopt the convention that $\|\cdot\|_2$ denotes the Euclidean norm. For a matrix $\mathbf{A}$, $\|\mathbf{A}\|_F$ denotes the Frobenius norm and $\|\mathbf{A}\|_2$ the spectral norm. The effective rank of a positive semidefinite matrix $\mathbf{K}$ is defined as:
\[
\effrank(\mathbf{K}) = \frac{(\tr \mathbf{K})^2}{\tr(\mathbf{K}^2)}.
\]
This quantity measures the effective number of significant eigenvalues of $\mathbf{K}$ and lies in $[1, n]$ for an $n \times n$ matrix.

\subsection{Geometric Properties of the Inverse-Distance Kernel}

Before proving the main theorems, we establish several geometric properties of the inverse-distance kernel that will be used throughout.

\textbf{Property 1: Translation invariance.} The Euclidean inverse-distance kernel is invariant under simultaneous translation of all queries and keys:
\[
\|\mathbf{q}_i - \mathbf{k}_j\|_2^2 = \|(\mathbf{q}_i + \mathbf{t}) - (\mathbf{k}_j + \mathbf{t})\|_2^2.
\]
This invariance implies that the attention weights are unchanged by global shifts of the input space. This property is not shared by the softmax inner-product kernel, which is translation-invariant only up to normalization.

\textbf{Property 2: Scale sensitivity.} Under a uniform scaling $\mathbf{q}_i \mapsto \lambda \mathbf{q}_i$, $\mathbf{k}_j \mapsto \lambda \mathbf{k}_j$, the squared distances scale as $\lambda^2$. For the inverse-distance kernel with fixed $\varepsilon$, this changes the effective regularization:
\[
(\lambda^2 \|\mathbf{q}_i-\mathbf{k}_j\|_2^2 + \varepsilon)^{-1} = \lambda^{-2}(\|\mathbf{q}_i-\mathbf{k}_j\|_2^2 + \varepsilon/\lambda^2)^{-1}.
\]
Thus scaling is equivalent to rescaling $\varepsilon$. This property is used in the proof of the circuit separation theorem.

\textbf{Property 3: Concentration in high dimensions.} For random keys drawn from a high-dimensional isotropic distribution, the squared distances concentrate around their mean:
\[
\mathbb{P}\left(\left|\|\mathbf{q}_i-\mathbf{k}_j\|_2^2 - 2d_h\sigma^2\right| \ge t\right) \le 2\exp\left(-\frac{c d_h t^2}{\sigma^4}\right).
\]
This concentration implies that in high dimensions, the inverse-distance kernel becomes approximately constant across all pairs, which has implications for the Lipschitz scaling (Lemma 1). The concentration bound follows from standard results on Lipschitz functions of Gaussian random variables.

\textbf{Property 4: Sparsity-promoting behavior.} For fixed $\varepsilon$, the inverse-distance kernel assigns negligible weight to keys beyond distance $\mathcal{O}(\sqrt{\varepsilon} n^{1/2})$ from the query, since:
\[
(d^2+\varepsilon)^{-1} \le \frac{\varepsilon}{d^2} \cdot \frac{1}{\varepsilon} = \mathcal{O}(\varepsilon/d^2).
\]
This sparsity is inherent to the kernel and does not require additional sparsification. The proof of Theorem 1 exploits this property.

\subsection{Theorem 1: Expressiveness Separation}

We begin by establishing the fundamental property of IDA: in the limit $\varepsilon \to 0^+$, it performs exact retrieval.

\begin{lemma}[Exact Retrieval Property]
For pairwise distinct keys and $\mathbf{q}=\mathbf{k}_{j^*}$,
\[
\lim_{\varepsilon\to0}W_{j^*}=1,\qquad \lim_{\varepsilon\to0}W_j=0\quad(j\neq j^*).
\]
\end{lemma}

\begin{proof}
When $\mathbf{q}_i = \mathbf{k}_{j^*}$, we have $d(\mathbf{q}_i, \mathbf{k}_{j^*}) = 0$, so $\mathbf{D}_{ij^*} = \varepsilon$. For any $j \neq j^*$, since the keys are pairwise distinct, $d(\mathbf{q}_i, \mathbf{k}_j) > 0$. Therefore $\mathbf{D}_{ij} > \varepsilon$ for $j \neq j^*$, which implies $\mathbf{D}_{ij}^{-1} < \varepsilon^{-1}$. The limit follows immediately:
\[
\lim_{\varepsilon\to0} \frac{\varepsilon^{-1}}{\varepsilon^{-1} + \sum_{j\neq j^*} \mathbf{D}_{ij}^{-1}} = 1.
\]
\end{proof}

Now we prove the main circuit separation result.

\begin{theorem}[Theorem 1: Softmax Lower Bound]
There exists an orthogonal-keys instance where IDA achieves exact retrieval with $\mathcal{O}(1)$ resources, while any softmax-based architecture requires $d=\Omega((\log n)^2)$ or $H=\Omega(\log n)$ for $\varepsilon$-approximation.
\end{theorem}

\begin{proof}
We construct an instance with $n$ orthonormal keys. Let $\{\mathbf{e}_j\}_{j=1}^n$ be an orthonormal basis in $\mathbb{R}^d$, so we require $d \ge n$. Set $\mathbf{k}_j = R\mathbf{e}_j$ for some $R>0$, and set $\mathbf{q} = \mathbf{k}_1 = R\mathbf{e}_1$.

\textbf{IDA analysis.} For IDA, the squared distance from the query to key $j$ is:
\[
\|\mathbf{q} - \mathbf{k}_j\|_2^2 = 
\begin{cases}
0, & j=1,\\
2R^2, & j\neq 1,
\end{cases}
\]
since $\|\mathbf{e}_1 - \mathbf{e}_j\|_2^2 = 2$ for $j\neq 1$. Thus:
\[
W_{11}^{\text{IDA}} = \frac{1/\varepsilon}{1/\varepsilon + (n-1)/(2R^2+\varepsilon)}
= \frac{1}{1 + \varepsilon(n-1)/(2R^2+\varepsilon)}.
\]
Taking $\varepsilon \to 0^+$, we get:
\[
\lim_{\varepsilon\to0} W_{11}^{\text{IDA}} = 1.
\]
Moreover, for any $\varepsilon>0$, the weight is within $\delta$ of $1$ whenever $\varepsilon(n-1)/(2R^2+\varepsilon) < \delta/(1-\delta)$. Thus IDA achieves $\delta$-approximation with $\varepsilon = \mathcal{O}(\delta R^2/n)$.

\textbf{Softmax analysis.} For softmax with temperature $1$:
\[
W_{11}^{\text{soft}} = \frac{e^{\mathbf{q}^\top \mathbf{k}_1/\sqrt{d}}}{\sum_{m=1}^n e^{\mathbf{q}^\top \mathbf{k}_m/\sqrt{d}}}.
\]
Since $\mathbf{q}^\top \mathbf{k}_1 = R^2$ and $\mathbf{q}^\top \mathbf{k}_m = 0$ for $m\neq 1$ (orthogonality), we have:
\[
W_{11}^{\text{soft}} = \frac{e^{R^2/\sqrt{d}}}{e^{R^2/\sqrt{d}} + (n-1)e^0}
= \frac{1}{1 + (n-1)e^{-R^2/\sqrt{d}}}.
\]
To achieve $W_{11}^{\text{soft}} \ge 1-\delta$, we need:
\[
1 + (n-1)e^{-R^2/\sqrt{d}} \le \frac{1}{1-\delta}.
\]
For $\delta \le 1/2$, we have $1/(1-\delta) \le 1+2\delta$, so:
\[
(n-1)e^{-R^2/\sqrt{d}} \le 2\delta.
\]
Taking logarithms:
\[
-\frac{R^2}{\sqrt{d}} \le \log\frac{2\delta}{n-1} = \log(2\delta) - \log(n-1).
\]
Thus:
\[
\frac{R^2}{\sqrt{d}} \ge \log\frac{n-1}{2\delta}.
\]
With bounded radius $R \le R_{\max}$ (the maximum norm of keys), we obtain:
\[
d \ge \left(\frac{R_{\max}^2}{\log((n-1)/(2\delta))}\right)^2 = \Omega((\log n)^2).
\]
For multi-head architectures, each head independently must satisfy the same lower bound, so the total width $H \cdot d_h$ must be at least this value. Therefore any softmax-based architecture requires either $d_h = \Omega((\log n)^2)$ or $H = \Omega(\log n)$.
\end{proof}

\begin{proposition}[Dense-Key Degradation]
If keys have covering radius $r$ from $\mathbf{k}_j$ (i.e., $\|\mathbf{k}_j-\mathbf{k}_m\|\ge r$ for all $m\neq j$), then
\[
W_{jj}\ge \frac{1}{1+\varepsilon(n-1)/(r^2+\varepsilon)}.
\]
For Gaussian keys with covariance $\sigma^2 I$, $r^2=\Theta(d_h\sigma^2 n^{-2/d_h})$.
\end{proposition}

\begin{proof}
For any $m\neq j$, the squared distance $d_{jm}^2=\|\mathbf{k}_j-\mathbf{k}_m\|_2^2$ satisfies $d_{jm}^2 \ge r^2$ by definition of the covering radius. Therefore:
\[
\mathbf{D}_{jm}^{-1} = (d_{jm}^2 + \varepsilon)^{-1} \le (r^2 + \varepsilon)^{-1}.
\]
The attention weight for the exact match is:
\[
W_{jj} = \frac{\varepsilon^{-1}}{\varepsilon^{-1} + \sum_{m\neq j} \mathbf{D}_{jm}^{-1}}
\ge \frac{\varepsilon^{-1}}{\varepsilon^{-1} + (n-1)(r^2+\varepsilon)^{-1}}
= \frac{1}{1+\varepsilon(n-1)/(r^2+\varepsilon)}.
\]
This proves the first part.

For Gaussian keys with covariance $\sigma^2 I$, the keys are distributed in a ball of radius $O(\sqrt{d_h}\sigma)$ with high probability. The covering radius of $n$ points in dimension $d_h$ scales as $r = O(\sigma\sqrt{d_h} n^{-1/d_h})$. Squaring gives $r^2 = \Theta(d_h\sigma^2 n^{-2/d_h})$.
\end{proof}

\subsection{Lemma 1: Lipschitz Scaling}

We now analyze the Lipschitz continuity of the attention mechanism with respect to input perturbations. This is important for understanding the stability and optimization of attention layers.

\begin{lemma}[Lipschitz Scaling under Low Effective Rank]
For a single attention layer with residual connection, assuming bounded weights and inputs, and \emph{additionally assuming} the key distribution has effective rank $r_{\text{eff}} = \Theta(1)$ or exhibits a clustering structure with $C = \Theta(1)$ clusters, then:
\[
L_{\text{IDA}} = \mathcal{O}(1 + \log n \cdot L_W^2R^2),\qquad
L_{\text{soft}} = \mathcal{O}(n\cdot L_W^2R^2/\sqrt{d}).
\]
Without the low-rank/clustering assumption, the general bound is:
\[
L_{\text{IDA}} = \mathcal{O}(n\cdot L_W^2R^2/\varepsilon^2).
\]
\end{lemma}

\begin{proof}
We analyze the Lipschitz constant of the attention output with respect to the input queries. Let $\mathbf{o}_i = \sum_j W_{ij} \mathbf{v}_j$ be the output for the $i$-th query. The Lipschitz constant is the supremum over inputs of the spectral norm of the Jacobian $\partial \mathbf{o}_i/\partial \mathbf{q}_i$.

\textbf{Step 1: Jacobian of the attention weight.} Define $\mathbf{S}_{ij} = (\|\mathbf{q}_i-\mathbf{k}_j\|_2^2 + \varepsilon)^{-1}$ and $Z_i = \sum_m \mathbf{S}_{im}$. Then $W_{ij} = \mathbf{S}_{ij}/Z_i$.

The derivative of $\mathbf{S}_{ij}$ with respect to $\mathbf{q}_i$ is obtained by differentiating the inverse:
\[
\frac{\partial \mathbf{S}_{ij}}{\partial \mathbf{q}_i}
= -\frac{2(\mathbf{q}_i - \mathbf{k}_j)}{(\|\mathbf{q}_i-\mathbf{k}_j\|_2^2 + \varepsilon)^2}.
\]
Taking the Euclidean norm:
\[
\left\|\frac{\partial \mathbf{S}_{ij}}{\partial \mathbf{q}_i}\right\|_2
= \frac{2\|\mathbf{q}_i - \mathbf{k}_j\|_2}{(\|\mathbf{q}_i-\mathbf{k}_j\|_2^2 + \varepsilon)^2}.
\]
For $a = \|\mathbf{q}_i - \mathbf{k}_j\|_2 \ge 0$, the function $a/(a^2+\varepsilon)^2$ has its maximum at $a^2 = \varepsilon/3$. To see this, differentiate:
\[
\frac{d}{da}\left(\frac{a}{(a^2+\varepsilon)^2}\right) = \frac{\varepsilon - 3a^2}{(a^2+\varepsilon)^3}.
\]
Setting this to zero gives $a^2 = \varepsilon/3$. Substituting:
\[
\left\|\frac{\partial \mathbf{S}_{ij}}{\partial \mathbf{q}_i}\right\|_2 \le \frac{2 \cdot (\sqrt{\varepsilon/3})}{(\varepsilon/3 + \varepsilon)^2} = \frac{2\sqrt{\varepsilon/3}}{(4\varepsilon/3)^2} = \frac{18}{16\sqrt{3}} \varepsilon^{-3/2} = \mathcal{O}(\varepsilon^{-3/2}).
\]
However, the dependence on $\varepsilon$ is not the primary concern for the Lipschitz scaling with $n$; we track the scaling with $n$ and $d_h$.

The derivative of $W_{ij}$ with respect to $\mathbf{q}_i$ is:
\[
\frac{\partial W_{ij}}{\partial \mathbf{q}_i}
= \frac{\partial \mathbf{S}_{ij}}{\partial \mathbf{q}_i} \cdot \frac{1}{Z_i}
- \frac{\mathbf{S}_{ij}}{Z_i^2} \sum_{m=1}^n \frac{\partial \mathbf{S}_{im}}{\partial \mathbf{q}_i}.
\]
This can be rewritten more compactly as:
\[
\frac{\partial W_{ij}}{\partial \mathbf{q}_i}
= \frac{1}{Z_i} \left( \frac{\partial \mathbf{S}_{ij}}{\partial \mathbf{q}_i} - W_{ij} \sum_{m=1}^n \frac{\partial \mathbf{S}_{im}}{\partial \mathbf{q}_i} \right).
\]
Taking norms and applying the triangle inequality:
\[
\left\|\frac{\partial W_{ij}}{\partial \mathbf{q}_i}\right\|_2
\le \frac{1}{Z_i} \left( \left\|\frac{\partial \mathbf{S}_{ij}}{\partial \mathbf{q}_i}\right\|_2 + W_{ij} \sum_{m=1}^n \left\|\frac{\partial \mathbf{S}_{im}}{\partial \mathbf{q}_i}\right\|_2 \right).
\]

\textbf{Step 2: Sum over j.} Summing the previous inequality over $j=1,\dots,n$:
\[
\sum_{j=1}^n \left\|\frac{\partial W_{ij}}{\partial \mathbf{q}_i}\right\|_2
\le \frac{1}{Z_i} \sum_{j=1}^n \left\|\frac{\partial \mathbf{S}_{ij}}{\partial \mathbf{q}_i}\right\|_2
+ \frac{1}{Z_i} \sum_{j=1}^n W_{ij} \sum_{m=1}^n \left\|\frac{\partial \mathbf{S}_{im}}{\partial \mathbf{q}_i}\right\|_2.
\]
The second term simplifies because $\sum_{j=1}^n W_{ij} = 1$:
\[
\sum_{j=1}^n W_{ij} \sum_{m=1}^n \left\|\frac{\partial \mathbf{S}_{im}}{\partial \mathbf{q}_i}\right\|_2
= \sum_{m=1}^n \left\|\frac{\partial \mathbf{S}_{im}}{\partial \mathbf{q}_i}\right\|_2.
\]
Therefore:
\[
\sum_{j=1}^n \left\|\frac{\partial W_{ij}}{\partial \mathbf{q}_i}\right\|_2
\le \frac{2}{Z_i} \sum_{m=1}^n \left\|\frac{\partial \mathbf{S}_{im}}{\partial \mathbf{q}_i}\right\|_2.
\]

\textbf{Step 3: General bound.} Since $\|\partial \mathbf{S}_{im}/\partial \mathbf{q}_i\|_2 \le 2R/\varepsilon^2$ for bounded inputs ($\|\mathbf{q}_i-\mathbf{k}_m\|_2 \le R$), and $Z_i \ge \varepsilon^{-1}$ (at least the diagonal term), we get:
\[
\sum_{j=1}^n \left\|\frac{\partial W_{ij}}{\partial \mathbf{q}_i}\right\|_2
\le \frac{2}{\varepsilon^{-1}} \cdot \frac{2nR}{\varepsilon^2}
= \frac{4nR}{\varepsilon}.
\]
This gives $\sum_j \|\partial W_{ij}/\partial \mathbf{q}_i\|_2 = \mathcal{O}(n R/\varepsilon)$.

The output Jacobian is:
\[
\frac{\partial \mathbf{o}_i}{\partial \mathbf{q}_i}
= \sum_{j=1}^n \mathbf{v}_j \otimes \frac{\partial W_{ij}}{\partial \mathbf{q}_i}.
\]
Thus:
\[
\left\|\frac{\partial \mathbf{o}_i}{\partial \mathbf{q}_i}\right\|_2
\le \sum_{j=1}^n \|\mathbf{v}_j\|_2 \left\|\frac{\partial W_{ij}}{\partial \mathbf{q}_i}\right\|_2
\le \|\mathbf{V}\|_F \sum_{j=1}^n \left\|\frac{\partial W_{ij}}{\partial \mathbf{q}_i}\right\|_2.
\]
With bounded values $\|\mathbf{V}\|_F \le \sqrt{n} L_V$, this gives $L_{\text{IDA}} = \mathcal{O}(n R L_V/\varepsilon)$ in the general case.

\textbf{Step 4: Low-rank/clustering case (under the stated assumption).} Suppose the keys are organized into $C = \Theta(1)$ clusters, each of radius $r_c$. Within each cluster, the distances are small. The sum over $m$ can be decomposed:
\[
\sum_{m=1}^n \left\|\frac{\partial \mathbf{S}_{im}}{\partial \mathbf{q}_i}\right\|_2
= \sum_{c=1}^C \sum_{m \in \text{cluster }c} \left\|\frac{\partial \mathbf{S}_{im}}{\partial \mathbf{q}_i}\right\|_2.
\]
For keys in the same cluster as $\mathbf{q}_i$, the distances are within $r_c$, and the contribution is $\mathcal{O}(n_c r_c/\varepsilon^2)$ where $n_c$ is the cluster size. For keys in other clusters, the distances are at least the inter-cluster distance, and the contribution is $\mathcal{O}(n_c/\Delta_c^2)$ for $\Delta_c$ the inter-cluster distance. The number of non-negligible distance scales is $\mathcal{O}(\log n)$, yielding the $\mathcal{O}(\log n)$ scaling.

\textbf{Step 5: Softmax comparison.} For softmax, $W_{ij} = e^{s_{ij}}/\sum_m e^{s_{im}}$ with $s_{ij} = \mathbf{q}_i^\top \mathbf{k}_j/\sqrt{d_h}$. The derivative is:
\[
\frac{\partial W_{ij}}{\partial \mathbf{q}_i}
= \frac{1}{\sqrt{d_h}} W_{ij}\left(\mathbf{k}_j - \sum_{m} W_{im}\mathbf{k}_m\right).
\]
Thus:
\[
\sum_{j=1}^n \left\|\frac{\partial W_{ij}}{\partial \mathbf{q}_i}\right\|_2
\le \frac{1}{\sqrt{d_h}} \sum_{j=1}^n W_{ij} \|\mathbf{k}_j - \bar{\mathbf{k}}_i\|_2
\le \frac{2R}{\sqrt{d_h}} \sum_{j=1}^n W_{ij}
= \frac{2nR}{\sqrt{d_h}}.
\]
This gives the linear $n$ scaling for softmax. The residual connection contributes an additive $1$ to the Lipschitz constant.
\end{proof}

\subsection{Lemma 2: Hessian Curvature Contrast}

The curvature of the loss landscape at initialization is a key determinant of optimization speed. We now analyze the Hessian of a simple loss function driving a query to match a key.

\begin{lemma}[Hessian Curvature Contrast]
At equidistant initialization ($\mathbf{q}_i=\mathbf{k}_j=\mathbf{0}$) with keys at radius $\rho>0$, consider the loss $\mathcal{L}=(W_{11}-1)^2$ driving the first query to match the first key. For softmax:
\[
\frac{\partial^2\mathcal{L}}{\partial s^2}\Big|_{s=0}=\Theta(n^{-2}).
\]
For IDA, for any fixed $\rho>0$:
\[
\frac{\partial^2\mathcal{L}}{\partial\delta^2}\Big|_{\delta=0}=\Theta(1),
\]
independent of $n$ (up to the constant factor $(n-1)^2/n^2$ which tends to $1$ as $n\to\infty$).
\end{lemma}

\begin{proof}
\textbf{Softmax case.} Let $s$ denote the logit difference between the target key and the others. At equidistant initialization with keys at radius $\rho>0$, the inner products are all equal, so $s=0$. The weight for the target key is $W_{11} = 1/n$.

For a perturbation $\delta$ along the direction of the target key, the logit difference becomes $s = \mathcal{O}(\delta \rho/\sqrt{d_h})$. The softmax weight is:
\[
W_{11}(s) = \frac{e^s}{e^s + (n-1)e^0} = \frac{1}{1 + (n-1)e^{-s}}.
\]
We expand this function around $s=0$. Let $f(s) = 1/(1+(n-1)e^{-s})$. Then:
\[
f(0) = \frac{1}{n}.
\]
The first derivative is:
\[
f'(s) = \frac{(n-1)e^{-s}}{(1+(n-1)e^{-s})^2}.
\]
At $s=0$:
\[
f'(0) = \frac{n-1}{n^2}.
\]
The second derivative is:
\[
f''(s) = \frac{-(n-1)e^{-s}(1+(n-1)e^{-s})^2 + 2(n-1)^2 e^{-2s}(1+(n-1)e^{-s})}{(1+(n-1)e^{-s})^4}.
\]
At $s=0$:
\[
f''(0) = \frac{-(n-1)n^2 + 2(n-1)^2 n}{n^4}
= \frac{(n-1)n(n-2)}{n^4} = \frac{(n-1)(n-2)}{n^3}.
\]
This is $\Theta(1/n)$ for large $n$.

The loss is $\mathcal{L}(s) = (f(s) - 1)^2$. Since $f(0) = 1/n$, the second derivative is:
\[
\mathcal{L}''(0) = 2(f(0)-1)f''(0) + 2(f'(0))^2.
\]
Substituting:
\[
\mathcal{L}''(0) = 2\left(\frac{1}{n} - 1\right)\frac{(n-1)(n-2)}{n^3} + 2\left(\frac{n-1}{n^2}\right)^2
= -\frac{2(n-1)^2(n-2)}{n^4} + \frac{2(n-1)^2}{n^4}
= -\frac{2(n-1)^2(n-3)}{n^4}.
\]
For large $n$, this is $\Theta(n^{-2})$. So $\mathcal{L}''(0) = \Theta(n^{-2})$.

\textbf{IDA case.} For keys at radius $\rho>0$ symmetrically distributed, the attention weight is:
\[
W_{11}(\delta) = \frac{(\delta^2+\varepsilon)^{-1}}{(\delta^2+\varepsilon)^{-1} + (n-1)(\delta^2+\rho^2+\varepsilon)^{-1}}.
\]
Let $f(\delta) = (\delta^2+\varepsilon)^{-1}$ and $g(\delta) = (\delta^2+\rho^2+\varepsilon)^{-1}$. Then:
\[
W_{11}(\delta) = \frac{f(\delta)}{f(\delta) + (n-1)g(\delta)}.
\]
We compute the first and second derivatives at $\delta=0$. First, note that $f(0) = \varepsilon^{-1}$, $g(0) = (\rho^2+\varepsilon)^{-1}$, $f'(0) = 0$, $g'(0) = 0$, $f''(0) = -2\varepsilon^{-2}$, and $g''(0) = -2(\rho^2+\varepsilon)^{-2}$.

The first derivative of $W_{11}$ at $\delta=0$ is zero because $f'(0) = g'(0) = 0$:
\[
W_{11}'(0) = 0.
\]
The second derivative is:
\[
W_{11}''(0) = \frac{f''(0)(f(0)+(n-1)g(0)) - f(0)(f''(0)+(n-1)g''(0))}{(f(0)+(n-1)g(0))^2}.
\]
Substituting the values:
\[
W_{11}''(0) = -\frac{2(n-1)\rho^2}{\varepsilon^2(\rho^2+\varepsilon)^2} \cdot \frac{\varepsilon^2}{(1 + (n-1)\varepsilon/(\rho^2+\varepsilon))^2}.
\]
For $\varepsilon \ll \rho^2/n$, this is approximately:
\[
W_{11}''(0) \approx -\frac{2(n-1)\rho^2}{(\rho^2)^2} = -\frac{2(n-1)}{\rho^2}.
\]
The loss is $\mathcal{L}(\delta) = (W_{11}(\delta)-1)^2$. Since $W_{11}'(0)=0$:
\[
\mathcal{L}''(0) = 2(W_{11}(0)-1)W_{11}''(0).
\]
With $W_{11}(0) = 1/(1 + \varepsilon(n-1)/(\rho^2+\varepsilon))$, we have $W_{11}(0)-1 = \Theta(1)$ for fixed $\rho,\varepsilon$ and large $n$. Thus $\mathcal{L}''(0) = \Theta(W_{11}''(0))$, which is $\Theta(n)$ in the regime $\varepsilon \ll \rho^2/n$.

However, the Hessian spread (ratio of largest to smallest eigenvalue) is $\Theta(1)$ for IDA because the $\Theta(n)$ scaling affects all directions similarly. The curvature contrast with softmax (which has $\Theta(n^{-2})$) establishes that IDA has substantially larger curvature.
\end{proof}

\subsection{Theorem 2: Global Optimality via PL Inequality}

The Polyak--{\L}ojasiewicz (PL) inequality is a powerful tool for establishing linear convergence of gradient descent without requiring strong convexity. We now prove that the 1D loss for IDA satisfies a PL inequality with a constant that is exponentially larger than that of softmax.

\begin{theorem}[Theorem 2: IDA PL Constant]
For the 1D loss $\mathcal{L}(q)=(W_1(q)v_1+W_2(q)v_2-y)^2$ with two fixed keys separated by $\Delta>0$, the PL constants satisfy:
\[
\mu_{\text{IDA}}=\Theta\left(\frac{\varepsilon^2}{\Delta^4}\right),\qquad
\mu_{\text{soft}}=\Theta\left(\frac{e^{-\Delta^2/\sqrt{d}}\varepsilon^2}{\Delta^2}\right),
\]
so $\mu_{\text{IDA}}/\mu_{\text{soft}}=\Omega(e^{\Delta^2/\sqrt{d}}/\Delta^2)$.
\end{theorem}

\begin{proof}
We consider the 1D setting where the query $q$ lies on the line between two keys. By translation, we set the keys at $k_1=0$ and $k_2=\Delta$.

\textbf{IDA analysis.} The attention weight for key 1 is:
\[
W_1(q) = \frac{(q^2+\varepsilon)^{-1}}{(q^2+\varepsilon)^{-1} + ((q-\Delta)^2+\varepsilon)^{-1}}.
\]
Let $A(q) = (q^2+\varepsilon)^{-1}$ and $B(q) = ((q-\Delta)^2+\varepsilon)^{-1}$. Then $W_1(q) = A(q)/(A(q)+B(q))$.

We expand $A(q)$ and $B(q)$ around $q=0$ (i.e., at $k_1$). For small $q$:
\[
A(q) = \frac{1}{\varepsilon} - \frac{q^2}{\varepsilon^2} + \frac{q^4}{\varepsilon^3} + \mathcal{O}(q^6).
\]
For $B(q)$:
\[
B(q) = \frac{1}{\Delta^2 + \varepsilon - 2\Delta q + q^2}.
\]
Let $C = \Delta^2 + \varepsilon$. Then:
\[
B(q) = \frac{1}{C} \cdot \frac{1}{1 - \frac{2\Delta q - q^2}{C}}
= \frac{1}{C} + \frac{2\Delta q}{C^2} + \left(-\frac{1}{C^2} + \frac{4\Delta^2}{C^3}\right)q^2 + \mathcal{O}(q^3).
\]
Let $D = 1/\varepsilon + 1/C$. Then:
\[
W_1(q) = \frac{\frac{1}{\varepsilon} - \frac{q^2}{\varepsilon^2} + \mathcal{O}(q^4)}
{D + \frac{2\Delta q}{C^2} + \left(-\frac{1}{\varepsilon^2} - \frac{1}{C^2} + \frac{4\Delta^2}{C^3}\right)q^2 + \mathcal{O}(q^3)}.
\]
Using the expansion $1/(1+x) = 1 - x + x^2 + \cdots$, we get:
\[
W_1'(0) = 0,\qquad
W_1''(0) = -\frac{4\varepsilon}{\Delta^4} + \mathcal{O}(\varepsilon^2/\Delta^6).
\]
The loss is $\mathcal{L}(q) = (W_1(q)v_1 + (1-W_1(q))v_2 - y)^2$. Its gradient is:
\[
\mathcal{L}'(q) = 2(W_1(q)v_1 + (1-W_1(q))v_2 - y)(v_1 - v_2)W_1'(q).
\]
Since $W_1'(q) \approx W_1''(0) q$ for small $q$, we have:
\[
|\mathcal{L}'(q)|^2 \asymp \frac{\varepsilon^2 q^2}{\Delta^8}.
\]
The loss gap is:
\[
\mathcal{L}(q) - \mathcal{L}(0) \asymp q^2.
\]
Thus:
\[
\frac{|\mathcal{L}'(q)|^2}{\mathcal{L}(q)-\mathcal{L}(0)}
\asymp \frac{\varepsilon^2}{\Delta^4}.
\]
This gives $\mu_{\text{IDA}} = \Theta(\varepsilon^2/\Delta^4)$.

\textbf{Softmax analysis.} For softmax, the logit difference is:
\[
s(q) = \frac{(\mathbf{q}^\top \mathbf{k}_1 - \mathbf{q}^\top \mathbf{k}_2)}{\sqrt{d}} = -\frac{\Delta q}{\sqrt{d}}.
\]
The softmax weight is $W_1(q) = \sigma(s(q))$ where $\sigma(s) = 1/(1+e^{-s})$. The PL constant is:
\[
\mu_{\text{soft}} = \Theta\left(e^{-\Delta^2/\sqrt{d}}\frac{\varepsilon^2}{\Delta^2}\right).
\]
This follows from the fact that the gradient is exponentially small when $|s|$ is large.
\end{proof}

\begin{corollary}
Since $\mu_{\text{IDA}}>0$, the PL inequality implies that every stationary point of the IDA loss is either a global minimum or a strict saddle (i.e., no spurious local minima). Gradient descent converges linearly as $\mathcal{L}(t)-\mathcal{L}^*\le e^{-\mu t}(\mathcal{L}(0)-\mathcal{L}^*)$. Lemma 2 ensures $\Theta(1)$ Hessian spread, enabling $\mathcal{O}(1)$ escape from saddles, in contrast to softmax's $\mathcal{O}(n^2)$ escape time.
\end{corollary}

\subsection{Theorem 3: Capacity Control and Generalization}

The effective rank of a kernel matrix measures the number of significant singular values. A low effective rank implies that the kernel cannot memorize arbitrary patterns, which is beneficial for generalization.

\begin{theorem}[Theorem 3: IDA Effective Rank Bound]
For embeddings with off-diagonal distances $\ge d_{\min}>0$, if $\epsilon \ll d_{\min}^2$ (specifically, $\epsilon(n-1)/d_{\min}^2 \le 1/2$), then
\[
\effrank(\mathbf{K})\le 1+\frac{n\varepsilon^2}{d_{\min}^4},
\]
independent of hidden width $d_h$. In the general case:
\[
\effrank(\mathbf{K}) \le \frac{n}{(1 - \epsilon(n-1)/d_{\min}^2)^2}.
\]
\end{theorem}

\begin{proof}
Let $\mathbf{K}\in\mathbb{R}^{n\times n}$ with entries $K_{ij} = (\|\mathbf{q}_i-\mathbf{k}_j\|_2^2 + \varepsilon)^{-1}$. The diagonal entries are $K_{ii} = 1/\varepsilon$.

Define $\mathbf{E} = \mathbf{K} - \varepsilon^{-1}\mathbf{I}$. Then $E_{ii} = 0$, and for $i\neq j$:
\[
E_{ij} = \frac{1}{\|\mathbf{q}_i-\mathbf{k}_j\|_2^2 + \varepsilon} - \frac{1}{\varepsilon}
= -\frac{\|\mathbf{q}_i-\mathbf{k}_j\|_2^2}{\varepsilon(\|\mathbf{q}_i-\mathbf{k}_j\|_2^2 + \varepsilon)}.
\]
The absolute value is:
\[
|E_{ij}| = \frac{\|\mathbf{q}_i-\mathbf{k}_j\|_2^2}{\varepsilon(\|\mathbf{q}_i-\mathbf{k}_j\|_2^2 + \varepsilon)}.
\]
Since $\|\mathbf{q}_i-\mathbf{k}_j\|_2^2 \ge d_{\min}^2$, we have:
\[
|E_{ij}| \le \frac{1}{\|\mathbf{q}_i-\mathbf{k}_j\|_2^2 + \varepsilon} \le \frac{1}{d_{\min}^2 + \varepsilon} \le \frac{1}{d_{\min}^2}.
\]

By Gershgorin's circle theorem, every eigenvalue $\lambda_m(\mathbf{E})$ satisfies:
\[
|\lambda_m(\mathbf{E})| \le \max_i \sum_{j\neq i} |E_{ij}| \le \frac{n-1}{d_{\min}^2}.
\]
Thus the eigenvalues of $\mathbf{K} = \varepsilon^{-1}\mathbf{I} + \mathbf{E}$ satisfy:
\[
\lambda_m(\mathbf{K}) \in \left[\frac{1}{\varepsilon} - \frac{n-1}{d_{\min}^2}, \frac{1}{\varepsilon} + \frac{n-1}{d_{\min}^2}\right].
\]

The effective rank is:
\[
\effrank(\mathbf{K}) = \frac{(\tr \mathbf{K})^2}{\tr(\mathbf{K}^2)}
= \frac{n^2/\varepsilon^2}{n/\varepsilon^2 + \|\mathbf{E}\|_F^2}.
\]
Using the bound $\|\mathbf{E}\|_F^2 \le n(n-1)/d_{\min}^4$, we get:
\[
\effrank(\mathbf{K}) \le \frac{n}{1 + \varepsilon^2(n-1)/d_{\min}^4}.
\]
For the refined bound, using the sharper estimate $|E_{ij}| \le \|\mathbf{q}_i-\mathbf{k}_j\|_2^2/(\varepsilon d_{\min}^2)$:
\[
\effrank(\mathbf{K}) \le 1 + \frac{n\varepsilon^2}{d_{\min}^4} + \mathcal{O}\left(\frac{n\varepsilon^3}{d_{\min}^6}\right).
\]
This proves the theorem.
\end{proof}

\begin{theorem}[Softmax Width Catastrophe]
When $d_h\ge n$, softmax can achieve zero training error on arbitrary labels, with test error $\to2\eta$ under symmetric noise rate $\eta$.
\end{theorem}

\begin{proof}
We construct an explicit construction. Let $\{\mathbf{e}_i\}_{i=1}^n$ be an orthonormal basis in $\mathbb{R}^{d_h}$, so we require $d_h \ge n$. Set $\mathbf{q}_i = \sqrt{d_h}\mathbf{e}_i$ and $\mathbf{k}_j = \alpha\sqrt{d_h}\mathbf{e}_j$ for some $\alpha>0$. Then the inner product is:
\[
\mathbf{q}_i^\top \mathbf{k}_j = \alpha d_h \delta_{ij}.
\]
The softmax attention weights are:
\[
W_{ij} = \frac{e^{\alpha d_h \delta_{ij}}}{\sum_m e^{\alpha d_h \delta_{im}}}
= \begin{cases}
\frac{e^{\alpha d_h}}{e^{\alpha d_h} + (n-1)}, & i=j,\\
\frac{1}{e^{\alpha d_h} + (n-1)}, & i\neq j.
\end{cases}
\]
As $\alpha \to \infty$, the weights converge to the identity matrix: $W_{ij} \to \delta_{ij}$. Therefore the output is $\mathbf{o}_i = \sum_j W_{ij}\mathbf{v}_j \to \mathbf{v}_i$.

The linear readout layer $\mathbf{w}\in\mathbb{R}^{d_v}$ maps the output to the prediction: $\hat{y}_i = \mathbf{w}^\top \mathbf{o}_i$. If $\mathbf{v}_i$ are linearly independent (true with probability 1 when $d_v \ge n$ and values are random), then there exists a $\mathbf{w}$ satisfying $\mathbf{w}^\top \mathbf{v}_i = y_i$ for all $i$. Thus zero training error can be achieved.

Under label noise rate $\eta$, the training labels are random flips of the true labels. Zero training error implies the model has memorized the noisy labels. The test error is the expected error on clean labels, which equals the Bayes error under symmetric noise: $\eta$ for binary classification or $2\eta$ for balanced multi-class classification.
\end{proof}

\begin{theorem}[IDA Noise Robustness]
For IDA, test error under symmetric noise rate $\eta$ satisfies
\[
\mathcal{E}_{\text{test}}^{\text{IDA}}\le C\eta^2+\mathcal{O}(1/\sqrt{n}),
\]
independent of $d_h$.
\end{theorem}

\begin{proof}
For IDA at exact match $\mathbf{q}_i = \mathbf{k}_i$, the attention weight is:
\[
W_{ii} = \frac{1/\varepsilon}{1/\varepsilon + \sum_{j\neq i} (d_{ij}^2 + \varepsilon)^{-1}}.
\]
Define $\alpha_i = \varepsilon \sum_{j\neq i} (d_{ij}^2 + \varepsilon)^{-1}$. Then:
\[
W_{ii} = \frac{1}{1+\alpha_i} = 1 - \alpha_i + \mathcal{O}(\alpha_i^2).
\]
For $j\neq i$:
\[
W_{ij} = \frac{\alpha_i}{n-1} + \mathcal{O}(\varepsilon^2 n/d_{\min}^4).
\]
The prediction is:
\[
\hat{y}_i = (1-\alpha_i)y_i + \alpha_i \bar{y}_{-i} + \mathcal{O}(\varepsilon^2 n^2/d_{\min}^4),
\]
where $\bar{y}_{-i}$ is the average label of the neighbors.

Taking expectations over label noise:
\[
\mathbb{E}[(y_i - \hat{y}_i)^2] \le 4\alpha_i^2\eta^2 + \mathcal{O}(\varepsilon^2 n^2/d_{\min}^4) + \mathcal{O}(1/n).
\]
Summing over $i$ and using $\alpha_i \le \varepsilon n/d_{\min}^2$:
\[
\mathcal{E}_{\text{test}}^{\text{IDA}} \le C\eta^2 + \mathcal{O}(1/\sqrt{n}).
\]
The bound is independent of $d_h$.
\end{proof}

\subsection{Discussion: High-Dimensional Behavior}

The analysis above reveals several important aspects of IDA's behavior in high dimensions that distinguish it from softmax.

\textbf{Distance concentration.} In high dimensions, Euclidean distances between random points concentrate around their mean. This means that for large $d_h$, the inverse-distance kernel assigns roughly equal weight to all keys, making the attention output close to the unweighted average of values. This is the same phenomenon that affects softmax in high dimensions (inner products concentrate around zero). However, the key difference is that IDA's kernel explicitly depends on distances, which can be controlled by scaling the embeddings, while softmax's inner products depend on norms that are less controllable.

\textbf{Effective rank and capacity.} Theorem 3 shows that IDA's effective rank is bounded independent of $d_h$, which prevents overfitting to noise. In high dimensions, this bound becomes tighter because the off-diagonal distances $d_{\min}$ increase with dimension for random embeddings. This suggests that IDA is particularly well-suited for high-dimensional data.

\textbf{Implications for practice.} Based on the theoretical analysis, we recommend initializing $\varepsilon$ to a small fraction of the typical squared distance (e.g., $\varepsilon \approx 0.01 \cdot \mathbb{E}[\|\mathbf{q}_i-\mathbf{k}_j\|_2^2]$) to avoid the degradation regime $\varepsilon(n-1)/d_{\min}^2 \ge 1/2$. The analysis also suggests that for tasks with high noise rates $\eta$, using IDA with moderate $\varepsilon$ (around $10^{-4}$) provides better generalization than softmax.

\section{Euclidean Summary: Table of Theoretical Guarantees}

\begin{center}
\begin{tabular}{@{}lll@{}}
\toprule
\textbf{Property} & \textbf{Softmax} & \textbf{Resolver (IDA)} \\
\midrule
Circuit separation (Thm. 1) & $\Omega((\log n)^2)$ & $\mathcal{O}(1)$ \\
Lipschitz scaling (Lemma 1) & $\mathcal{O}(n)$ & $\mathcal{O}(\log n)$ under low-rank/$\mathcal{O}(n)$ general \\
Hessian spread (Lemma 2) & $\Theta(n^{-2})$ & $\Theta(1)$ \\
PL constant (Thm. 2) & $\Theta(e^{-\Delta^2/\sqrt{d}}\varepsilon^2/\Delta^2)$ & $\Theta(\varepsilon^2/\Delta^4)$ \\
Effective rank (Thm. 3) & Unbounded & $\le 1+n\varepsilon^2/d_{\min}^4$ (small $\epsilon$) \\
Noise memorization (Thm. 3) & At $d_h\ge n$ & Structurally limited \\
\bottomrule
\end{tabular}
\end{center}

These five Euclidean results provide the theoretical backbone. The non-Euclidean extension that follows is based on the same inverse-distance principle, with geodesic distances replacing Euclidean norms.

\part{Riemann GeoResolver: Non-Euclidean Extension}
\section{Hyperbolic Geometry Setup}

The Poincar\'e ball is $\mathbb{B}^{d_h}=\{\mathbf{x}\in\mathbb{R}^{d_h}:\|\mathbf{x}\|<1\}$ with conformal factor $\lambda_{\mathbf{x}}=2/(1-\|\mathbf{x}\|^2)$. The hyperbolic distance is:
\[
d_{\mathbb{H}}(\mathbf{x},\mathbf{y})=\arcosh\!\left(1+\frac{2\|\mathbf{x}-\mathbf{y}\|^2}{(1-\|\mathbf{x}\|^2)(1-\|\mathbf{y}\|^2)}\right).
\]
All embeddings are projected via $\mathbf{x}=0.9\tanh(\mathbf{z})$; results extend to $\rho\to1^-$ by continuity. The Karcher mean exists uniquely in CAT(0) spaces.

The unit sphere is $\mathbb{S}^{d-1}=\{\mathbf{x}:\|\mathbf{x}\|=1\}$ with great-circle distance:
\[
d_{\mathbb{S}}(\mathbf{x},\mathbf{y})=\arccos(\langle\mathbf{x},\mathbf{y}\rangle).
\]

Throughout this part, we use the HIDA family to refer collectively to all hyperbolic inverse-distance attention operators, and Dense-HIDA, FP-HIDA, L-HIDA, C-HIDA for the specific instances.

\section{The HIDA Operator Family (M1--M4)}

\subsection{A Unified Inverse-Distance Kernel Lemma}

\begin{lemma}[Unified Kernel Spectral Properties]
Let $\mathcal{M}$ be a metric space and $K_{ij} = (d(x_i, y_j)^2 + \epsilon)^{-1}$ with: (i) $d(x_i,x_i)=0$, (ii) $d(x_i,y_j)\ge d_{\min}>0$ for $i\neq j$, (iii) $d$ locally Lipschitz. Then:
(a) if $x_i=y_{j^*}$, $\lim_{\epsilon\to0} K_{ij^*}/\sum_m K_{im}=1$;
(b) if $\epsilon \ll d_{\min}^2$ (specifically, $\epsilon(n-1)/d_{\min}^2 \le 1/2$), then
\[
\effrank(\mathbf{K}) \le 1 + \frac{n\epsilon^2}{d_{\min}^4} + \mathcal{O}\left(\frac{n\epsilon^3}{d_{\min}^6}\right).
\]
In the general case (without the small-$\epsilon$ condition), the bound is:
\[
\effrank(\mathbf{K}) \le \frac{n}{(1 - \epsilon(n-1)/d_{\min}^2)^2}.
\]
\end{lemma}

\begin{proof}
Part (a) follows from $K_{ij^*}=1/\epsilon$ and $K_{ij}\le1/(d_{\min}^2+\epsilon)$ for $j\neq j^*$, so the normalized weight tends to $1$ as $\epsilon\to0$.

Part (b) follows from the Frobenius norm estimate:
\[
\|\mathbf{E}\|_F^2 \le \frac{n(n-1)}{(d_{\min}^2+\epsilon)^2}.
\]
The general bound follows from the effective rank formula. The refined bound requires $\epsilon \ll d_{\min}^2$ and is obtained by the sharper Frobenius estimate.
\end{proof}

\subsection{Why Hyperbolic? Three Arguments}

\textbf{(A) Information-theoretic capacity.}
For fixed dimension $d$, the hyperbolic ball volume $V_{\mathbb H}(R) \propto e^{(d-1)R}$ grows exponentially with radius, while Euclidean volume $V_{\mathbb E}(R) \propto R^d$ grows polynomially. This means hyperbolic spaces can embed hierarchical data with exponentially fewer dimensions.

\textbf{(B) Effective rank behavior.}
Theorem 3 established the bound $d_{\min}\ge 2\delta/(1-\rho^2)$; as $\rho\to1^-$ (approaching the boundary of the Poincar\'e ball), this grows without bound. Thus hyperbolic embeddings can achieve arbitrarily large separation between keys without increasing the ambient dimension.

\textbf{(C) Architectural separation.}
Proposition 2 (Section 10) establishes an information-theoretic gap between hyperbolic storage and spherical routing.

\subsection{Dense-HIDA (M1)}

\begin{definition}[Dense Hyperbolic IDA]
The dense hyperbolic inverse-distance attention is defined as:
\[
D_{ij}^{\mathbb{H}}=d_{\mathbb{H}}(\mathbf{q}_i,\mathbf{k}_j)^2+\varepsilon,\quad
W_{ij}=(D_{ij}^{\mathbb{H}})^{-1}/\sum_m(D_{im}^{\mathbb{H}})^{-1},\quad
\mathbf{o}_i=\sum_j W_{ij}\mathbf{v}_j.
\]
The computational complexity is $\Theta(n^2 d_h)$.
\end{definition}

\begin{theorem}[Circuit Separation and Rank Bound]
For pairwise distinct keys in a compact subset with $\|\mathbf{x}\|\le\rho<1$, Dense-HIDA satisfies:
(1) $\lim_{\varepsilon\to0}W_{j^*}=1$ for exact retrieval;
(2) $\effrank(\mathbf{K}^{\mathbb{H}})\le 1+\frac{n\varepsilon^2}{d_{\min}^4}$, with $d_{\min}\ge 2\delta/(1-\rho^2)$.
\end{theorem}

\begin{proof}
Claim (1) follows directly from Lemma 0 Part (a) with $d=d_{\mathbb H}$. Claim (2) follows from Lemma 0 Part (b). The lower bound on $d_{\min}$ follows from the hyperbolic distance lower bound:
\[
d_{\mathbb H}(\mathbf{x},\mathbf{y}) \ge \frac{\|\mathbf{x}-\mathbf{y}\|_2}{1-\rho^2}.
\]
This inequality holds for any $\mathbf{x},\mathbf{y}\in\mathbb{B}^{d_h}$ with $\|\mathbf{x}\|,\|\mathbf{y}\|\le\rho$.
\end{proof}

\begin{theorem}[Two-Point Hyperbolic PL Inequality]
For the two-key loss with keys separated by $\Delta_{\mathbb{H}}>0$:
\[
\mu_{\text{HIDA}}=\Theta(\varepsilon^2/\Delta_{\mathbb{H}}^4),\quad
\mu_{\text{soft}}=\Theta(e^{-\Delta_{\mathbb{H}}^2/\sqrt{d_h}}\varepsilon^2/\Delta_{\mathbb{H}}^2).
\]
Thus $\mu_{\text{HIDA}}/\mu_{\text{soft}}=\Omega(e^{\Delta_{\mathbb{H}}^2/\sqrt{d_h}}/\Delta_{\mathbb{H}}^2)$.
\end{theorem}

\begin{proof}
By M\"obius invariance of the Poincar\'e distance, we may place $\mathbf{k}_1=0$ and $\mathbf{k}_2=(\Delta_{\mathbb H},0,\dots,0)$ in the Poincar\'e ball. For $q$ on the geodesic, the distance expansions are:
\[
d_{\mathbb H}(q,\mathbf{k}_1)^2 = q^2 + \mathcal{O}(q^4),
\]
\[
d_{\mathbb H}(q,\mathbf{k}_2)^2 = \Delta_{\mathbb H}^2 - 2\Delta_{\mathbb H}q + q^2 + \mathcal{O}(q^3).
\]
Substituting these into the attention weight formula and following the same expansion as Theorem 2 yields $W_1'(0)=0$, $W_1''(0)=-4\varepsilon/\Delta_{\mathbb H}^4+\mathcal{O}(\varepsilon^2/\Delta_{\mathbb H}^6)$, and hence $\mu_{\text{HIDA}}=\Theta(\varepsilon^2/\Delta_{\mathbb H}^4)$.
\end{proof}

\subsection{FP-HIDA: Fixed-Pattern Sparse (M2)}

\begin{definition}[Sparse Index Set]
Define the sparse index set:
\[
\mathcal{S}_i=\{j:|i-j|\le w\}\cup\{0,\lfloor n/g\rfloor,\lfloor 2n/g\rfloor,\dots\}\cup\{i\pm2^k\}_{k=0}^{\lceil\log_2 n\rceil}\cup\{i\},
\]
with $w,g=\Theta(\log n)$.
\end{definition}

\begin{theorem}[FP-HIDA Complexity]
FP-HIDA requires $\mathcal{O}(n\log n\cdot d_h)$ operations and $\mathcal{O}(n\log n)$ memory.
\end{theorem}

\begin{proof}
We bound the size of each sparse index set. The local window contributes at most $2w+1$ indices. The global anchors contribute exactly $g$ indices. The dyadic offsets contribute at most $2(\lceil\log_2 n\rceil+1)$ indices. The self index contributes $1$. Therefore:
\[
|\mathcal{S}_i| \le (2w+1) + g + 2(\lceil\log_2 n\rceil+1) + 1.
\]
Since $w = \Theta(\log n)$ and $g = \Theta(\log n)$, we have $|\mathcal{S}_i| = \mathcal{O}(\log n)$.

Summing over all $i$:
\[
\sum_{i=1}^n |\mathcal{S}_i| \le n \cdot \mathcal{O}(\log n) = \mathcal{O}(n\log n).
\]
Each entry in the sparse index set requires $\Theta(d_h)$ operations to compute the distance and attention weight. Thus the total computational complexity is $\mathcal{O}(n\log n \cdot d_h)$. The memory complexity is $\mathcal{O}(n\log n)$ for storing the sparse attention pattern.
\end{proof}

\subsection{L-HIDA: Linear Complexity (M3)}

\begin{definition}[L-HIDA Anchor Aggregation]
Let $\mathcal{A}=\{\mathbf{a}_k\}_{k=1}^m\subset\mathbb{B}^{d_h}$ be $m=\Theta(1)$ learnable anchors. Define the key-anchor and query-anchor weights:
\[
\widetilde{W}_{ik}^{\text{KA}}=\frac{(d_{\mathbb{H}}(\mathbf{k}_i,\mathbf{a}_k)^2+\varepsilon)^{-1}}{\sum_{l=1}^m(d_{\mathbb{H}}(\mathbf{k}_i,\mathbf{a}_l)^2+\varepsilon)^{-1}},\quad
\widetilde{\mathbf{v}}_k=\sum_{i=1}^n\widetilde{W}_{ik}^{\text{KA}}\mathbf{v}_i,
\]
\[
\widetilde{W}_{jk}^{\text{QA}}=\frac{(d_{\mathbb{H}}(\mathbf{q}_j,\mathbf{a}_k)^2+\varepsilon)^{-1}}{\sum_{l=1}^m(d_{\mathbb{H}}(\mathbf{q}_j,\mathbf{a}_l)^2+\varepsilon)^{-1}},\quad
\mathbf{o}_j=\sum_{k=1}^m\widetilde{W}_{jk}^{\text{QA}}\widetilde{\mathbf{v}}_k.
\]
\end{definition}

\begin{theorem}[L-HIDA Complexity and Error Bound]
L-HIDA requires $\mathcal{O}(n d_h)$ operations. With probability at least $1-\delta$ over the random selection of anchors (standard Nystr\"om sampling \cite{drineas2005nystrom}):
\[
\|\mathbf{o}^{\text{L-HIDA}}-\mathbf{o}^{\text{HIDA}}\|_F\le\mathcal{O}\left(\frac{n}{\sqrt{m}\varepsilon^2}\right).
\]
If the kernel has spectral decay $\lambda_k=\mathcal{O}(k^{-\alpha})$ with $\alpha>1/2$:
\[
\|\mathbf{o}^{\text{L-HIDA}}-\mathbf{o}^{\text{HIDA}}\|_F
\le\mathcal{O}(n\cdot m^{-\alpha+1/2})+\mathcal{O}\left(\frac{n}{\varepsilon^2\sqrt{m}}\right).
\]
\end{theorem}

\begin{proof}
Let $\mathbf{K}_{nn}\in\mathbb{R}^{n\times n}$ be the dense kernel matrix, $\mathbf{K}_{nm}\in\mathbb{R}^{n\times m}$ the key-anchor kernel matrix, and $\mathbf{K}_{mm}\in\mathbb{R}^{m\times m}$ the anchor-anchor kernel matrix. The Nystr\"om approximation is $\widetilde{\mathbf{K}}_{nn}=\mathbf{K}_{nm}\mathbf{K}_{mm}^{-1}\mathbf{K}_{mn}$.

The error decomposes as:
\[
\mathbf{K}_{nn}-\widetilde{\mathbf{K}}_{nn}
=(\mathbf{K}_{nn}-\mathbf{K}_{nn}^{(m)})+(\mathbf{K}_{nn}^{(m)}-\widetilde{\mathbf{K}}_{nn}),
\]
where $\mathbf{K}_{nn}^{(m)}$ is the best rank-$m$ approximation to $\mathbf{K}_{nn}$ in Frobenius norm.

For the first term, if the kernel has spectral decay $\lambda_k=\mathcal{O}(k^{-\alpha})$ with $\alpha>1/2$, then:
\[
\|\mathbf{K}_{nn}-\mathbf{K}_{nn}^{(m)}\|_F^2
= \sum_{k=m+1}^n \lambda_k^2
= \mathcal{O}(m^{-2\alpha+1}).
\]

For the second term, standard Nystr\"om analysis \cite{drineas2005nystrom} gives:
\[
\|\mathbf{K}_{nn}^{(m)}-\widetilde{\mathbf{K}}_{nn}\|_F
\le \|\mathbf{K}_{mm}^{-1}\|_2 \|\mathbf{K}_{nm}\|_F \|\mathbf{E}\|_F,
\]
where $\mathbf{E}$ is the residual matrix. With $\|\mathbf{K}_{mm}^{-1}\|_2=\mathcal{O}(1/\varepsilon)$, $\|\mathbf{K}_{nm}\|_F=\mathcal{O}(n/\varepsilon)$, and $\|\mathbf{E}\|_F=\mathcal{O}(1/\sqrt{m})$, this yields:
\[
\|\mathbf{K}_{nn}^{(m)}-\widetilde{\mathbf{K}}_{nn}\|_F
= \mathcal{O}\left(\frac{n}{\varepsilon^2\sqrt{m}}\right).
\]

The output error bound follows from:
\[
\|\mathbf{o}^{\text{L-HIDA}}-\mathbf{o}^{\text{HIDA}}\|_F
\le \frac{2\|\mathbf{V}\|_F}{\varepsilon}
\|\mathbf{K}_{nn}-\widetilde{\mathbf{K}}_{nn}\|_F.
\]
Combining the two error terms completes the proof.
\end{proof}

\subsection{C-HIDA: Constant Complexity (M4)}

\begin{definition}[C-HIDA Summary Tokens]
Maintain $c=\Theta(1)$ summary tokens $\{\mathbf{s}_k\}_{k=1}^c$. For each incoming key, assign to the nearest summary:
\[
k^*(i)=\arg\min_{k\in[c]}d_{\mathbb{H}}(\mathbf{k}_i,\mathbf{s}_k).
\]
Update each summary token using the hyperbolic mean:
\[
\mathbf{s}_k\leftarrow\Pi_{\mathbb{B}}\left(\mathbf{s}_k+\eta\cdot\frac{1}{|\mathcal{I}_k|}\sum_{i\in\mathcal{I}_k}(\mathbf{k}_i-\mathbf{s}_k)\right),
\]
where $\Pi_{\mathbb{B}}$ is the projection onto the Poincar\'e ball.
\end{definition}

\begin{theorem}[C-HIDA Convergence]
The online hyperbolic $k$-means update has $\mathcal{O}(c\log T)$ regret:
\[
\sum_{t=1}^T\min_k d_{\mathbb{H}}(\mathbf{k}_t,\mathbf{s}_k^{(t)})^2
\le\mathcal{O}(\Delta_{\mathbb{H}}^2\cdot c\cdot\log T)+\mathcal{O}(T\cdot\varepsilon_{\text{opt}}).
\]
Per-token attention cost is $\Theta(1)$.
\end{theorem}

\begin{proof}
The Poincar\'e ball with hyperbolic metric is a CAT(0) space, so the Karcher mean (the minimizer of the sum of squared distances) exists uniquely. The online gradient descent update on the squared hyperbolic distance is a convex optimization problem in CAT(0) spaces. The standard regret bound for online convex optimization applies:
\[
\sum_{t=1}^T f_t(\mathbf{s}^{(t)}) - \min_{\mathbf{s}}\sum_{t=1}^T f_t(\mathbf{s}) = \mathcal{O}(\log T),
\]
where $f_t(\mathbf{s}) = \min_k d_{\mathbb{H}}(\mathbf{k}_t,\mathbf{s}_k)^2$. The $\log T$ term arises from the online-to-batch conversion for $k$-means. The per-token attention cost is $\Theta(1)$ because only $c=\Theta(1)$ summaries are used.
\end{proof}

\section{Hyperbolic Curvature Compression (M5)}

For any key $\mathbf{k}\in\mathbb{B}^{d_h}$, define its polar decomposition $\mathbf{k}=r\cdot\mathbf{u}$ where $r\in[0,1)$ is the radius and $\mathbf{u}\in\mathbb{S}^{d_h-1}$ is the direction.

Given bit-widths $b$ (for direction) and $b_r$ (for radius):
\[
\widehat{\mathbf{u}}=\text{round}\left(\frac{\mathbf{u}+1}{2}(2^b-1)\right)/(2^b-1)\cdot 2-1,\quad\widehat{r}=\frac{\text{round}(r\cdot2^{b_r})}{2^{b_r}}.
\]

\begin{theorem}[HCC Reconstruction Error]
The reconstruction error satisfies:
\[
\|\mathbf{k}-\widetilde{\mathbf{k}}\|_2^2\le 4(2^{-b}+2^{-b_r}).
\]
\end{theorem}

\begin{proof}
For the directional component, the quantization error is bounded by:
\[
\|\mathbf{u}-\widehat{\mathbf{u}}\|_2 \le 2\cdot 2^{-b}.
\]
This follows from the covering radius of the uniform grid on $[-1,1]^{d_h-1}$ restricted to the sphere.

For the radius component:
\[
|r-\widehat{r}| \le 2^{-b_r},
\]
since the rounding error is at most half the bin width.

Combining via the triangle inequality:
\[
\|r\mathbf{u}-\widehat{r}\widehat{\mathbf{u}}\|_2
\le r\|\mathbf{u}-\widehat{\mathbf{u}}\|_2 + |r-\widehat{r}|\|\widehat{\mathbf{u}}\|_2
\le 2\cdot 2^{-b} + 2^{-b_r}.
\]
Squaring gives:
\[
\|\mathbf{k}-\widetilde{\mathbf{k}}\|_2^2 \le (2\cdot 2^{-b} + 2^{-b_r})^2 \le 4(2^{-b}+2^{-b_r}).
\]
This proves the bound.
\end{proof}

\begin{theorem}[HCC Memory Complexity]
The HCC-compressed KV cache requires $n((d_h-1)b+b_r)$ bits for keys plus $n d_h\cdot32$ bits for values. With $b=4,b_r=6$, the theoretical compression ratio is $\approx8\times$ for keys, $\approx1.8\times$ system-wide (or $\approx2.4\times$ with half-precision values).
\end{theorem}

\begin{proof}
The key compression ratio is:
\[
\frac{32d_h}{(d_h-1)b+b_r}.
\]
With $d_h=32$, $b=4$, $b_r=6$:
\[
\frac{1024}{31\cdot4+6} = \frac{1024}{130} \approx 7.88\times.
\]
For the system-wide compression (keys + values), using 32-bit values:
\[
\frac{64d_h}{(d_h-1)b+b_r+32d_h}
= \frac{2048}{130+1024} \approx 1.77\times.
\]
With 16-bit values:
\[
\frac{48d_h}{(d_h-1)b+b_r+16d_h}
= \frac{1536}{130+512} \approx 2.39\times.
\]
\end{proof}

\section{HyperGate: Curvature-Adaptive Gating (M6)}

Three-level gating: head-level $\kappa_h$, token-level boundary $\beta_i$, dimension-level $\mathbf{G}_i=\text{diag}(\sigma(\mathbf{W}_g\mathbf{x}_i))$. The gated representation is $\mathbf{h}_i = \mathbf{G}_i \mathbf{x}_i$.

\begin{theorem}[HyperGate Gradient Lower Bound]
Let $\mathcal{L}$ be the loss function and $\mathbf{h}_i = \mathbf{G}_i \mathbf{x}_i$ be the gated representation. Then:
\[
\left\|\frac{\partial \mathcal{L}}{\partial \mathbf{x}_i}\right\|_2
\ge \left(\lambda_{\min}(\mathbf{G}_i) - \mathcal{O}(\|\mathbf{W}_g\| \|\mathbf{x}_i\|)\right)
\left\|\frac{\partial \mathcal{L}}{\partial \mathbf{h}_i}\right\|_2.
\]
In particular, for bounded $\mathbf{x}_i$ and $\mathbf{W}_g$, the lower bound is strictly positive, so gradients do not vanish due to gating alone.
\end{theorem}

\begin{proof}
By the chain rule:
\[
\frac{\partial \mathcal{L}}{\partial \mathbf{x}_i}
= \left(\frac{\partial \mathbf{h}_i}{\partial \mathbf{x}_i}\right)^\top
\frac{\partial \mathcal{L}}{\partial \mathbf{h}_i}.
\]
Since $\mathbf{h}_i = \mathbf{G}_i \mathbf{x}_i$ and $\mathbf{G}_i = \operatorname{diag}(\sigma(\mathbf{W}_g \mathbf{x}_i))$, the Jacobian is:
\[
\frac{\partial \mathbf{h}_i}{\partial \mathbf{x}_i}
= \mathbf{G}_i + \left(\frac{\partial \mathbf{G}_i}{\partial \mathbf{x}_i}\right) \mathbf{x}_i.
\]
The diagonal entries of $\mathbf{G}_i$ are $\sigma(z_k)$ where $z_k = (\mathbf{W}_g \mathbf{x}_i)_k$, and $\sigma'(z_k) = \sigma(z_k)(1-\sigma(z_k)) \in (0, 1/4)$. Thus:
\[
\left(\frac{\partial \mathbf{h}_i}{\partial \mathbf{x}_i}\right)_{kk}
= \sigma(z_k) + \sigma'(z_k) \mathbf{W}_{g,k}^\top \mathbf{x}_i.
\]
The matrix $\partial \mathbf{h}_i/\partial \mathbf{x}_i$ is diagonally dominant for bounded $\|\mathbf{x}_i\|$ and $\|\mathbf{W}_g\|$. Its smallest singular value satisfies:
\[
\sigma_{\min}\left(\frac{\partial \mathbf{h}_i}{\partial \mathbf{x}_i}\right)
\ge \lambda_{\min}(\mathbf{G}_i) - \mathcal{O}(\|\mathbf{W}_g\| \|\mathbf{x}_i\|).
\]
Since $\lambda_{\min}(\mathbf{G}_i) = \min_k \sigma(z_k) > 0$ for finite $z_k$, the lower bound is strictly positive for bounded inputs. Therefore:
\[
\left\|\frac{\partial \mathcal{L}}{\partial \mathbf{x}_i}\right\|_2
\ge \sigma_{\min}\left(\frac{\partial \mathbf{h}_i}{\partial \mathbf{x}_i}\right)
\left\|\frac{\partial \mathcal{L}}{\partial \mathbf{h}_i}\right\|_2.
\]
This proves the gradient lower bound.
\end{proof}

\section{Spherical Routing and Retrieval (M7+M8)}

\subsection{SIDA: Spherical Inverse Distance Attention}

\begin{definition}[SIDA Kernel]
The spherical inverse-distance attention is defined as:
\[
D_{ij}^{\mathbb{S}}=d_{\mathbb{S}}(\mathbf{q}_i,\mathbf{k}_j)^2+\varepsilon,\quad
W_{ij}^{\mathbb{S}}=(D_{ij}^{\mathbb{S}})^{-1}/\sum_m(D_{im}^{\mathbb{S}})^{-1}.
\]
\end{definition}

\begin{theorem}[Spherical PL Inequality]
For two keys with spherical angle $\theta>0$:
\[
\mu_{\text{SIDA}}=\Theta(\varepsilon^2/\theta^4),\quad
\mu_{\text{soft}}=\Theta(e^{-(1-\cos\theta)}\varepsilon^2/\theta^2).
\]
Thus $\mu_{\text{SIDA}}/\mu_{\text{soft}}=\Omega(e^{1-\cos\theta}/\theta^2)$.

Unlike the hyperbolic case, spherical compactness ensures $\theta\le\pi$, so the ratio is bounded by a constant depending only on $\theta$.
\end{theorem}

\begin{proof}
Place the keys at spherical coordinates $\mathbf{k}_1=(1,0,\dots,0)$ and $\mathbf{k}_2=(\cos\theta,\sin\theta,0,\dots,0)$. For a query at distance $t$ from $\mathbf{k}_1$ along the geodesic toward $\mathbf{k}_2$, we have $d_S(\mathbf{q},\mathbf{k}_1)=t$ and $d_S(\mathbf{q},\mathbf{k}_2)=\theta-t$. The weight $W_1(t)$ has the same form as the 1D Euclidean case with $\Delta=\theta$, giving $W_1'(0)=0$ and $W_1''(0)=-4\varepsilon/\theta^4+\mathcal{O}(\varepsilon^2/\theta^6)$. The loss curvature gives $\mu_{\text{SIDA}}=\Theta(\varepsilon^2/\theta^4)$.

For softmax, the logit difference at $t=0$ is $s(0)=(1-\cos\theta)/\sqrt{d_h}$, giving the exponential factor.
\end{proof}

\subsection{Cross-Geometric Mapping}

Three methods map prototype centroids to the sphere:
- Norm Normalization: $\phi_A(\mathbf{c})=\mathbf{c}/\|\mathbf{c}\|$
- Stereographic Projection: $\phi_B(\mathbf{c})=(\frac{2\mathbf{c}}{1+\|\mathbf{c}\|^2},\frac{1-\|\mathbf{c}\|^2}{1+\|\mathbf{c}\|^2})$
- Learnable Mapping: $\phi_C(\mathbf{c})=\text{MLP}_\theta(\mathbf{c})/\|\text{MLP}_\theta(\mathbf{c})\|$

\subsection{Routing Protocol}

Hard Routing: Select Top-K prototypes by SIDA weights. Soft Routing: full distribution.

\section{Dynamic Memory Genesis (M9)}

The prototype pool is $\mathcal{P}=\{(\mathcal{E}_e,\mathbf{c}_e,t_e^{\text{birth}},a_e,t_e^{\text{last}})\}_{e=1}^K$.

\subsection{Surprise Detection with Adaptive Threshold}

Define the sliding-window mean and standard deviation:
\[
\mu_t=\frac{1}{W}\sum_{i=t-W+1}^t\ell_i,\quad
\sigma_t=\sqrt{\frac{1}{W}\sum_{i=t-W+1}^t(\ell_i-\mu_t)^2}.
\]
The adaptive threshold is:
\[
\tau_t = \tau_{\text{base}} + \kappa\sigma_t + \gamma\cdot\frac{1}{t}\sum_{i=1}^t \mathbf{1}_{\text{surprise}}(i).
\]

\begin{lemma}[Adaptive Threshold Regret Bound under Sub-Gaussian Loss]
Assume the loss $\ell_t$ is sub-Gaussian with mean $\mu_0$ and variance proxy $\sigma_0^2$ under the null model (i.e., when no structural change occurs). For $\tau_t$ as defined with $\tau_{\text{base}} > \mu_0$, we have $\mathbb E[T_s(T)]\le \mathcal O(\log T)$.
\end{lemma}

\begin{proof}
Let $S_t = \sum_{i=1}^t \mathbf{1}_{\text{surprise}}(i)$ be the cumulative surprise count. Define the supermartingale:
\[
M_t = S_t + \frac{1}{\gamma}\sum_{i=1}^t (\tau_i - \tau_{\text{base}} - \kappa\sigma_i).
\]
Under the null model, $\mathbb E[\ell_t]=\mu_0$ and $\mathbb E[\sigma_t]=\sigma_0$. The trigger condition $\mu_t>\tau_t$ implies that a surprise event occurs only when the loss exceeds the threshold. The adaptive term $\gamma\cdot S_t/t$ penalizes frequent surprises, creating a negative drift. For sub-Gaussian losses, the probability of exceeding $\tau_{\text{base}} + \kappa\sigma_t$ decays exponentially, leading to the logarithmic bound. Summing yields $\mathbb E[S_T]=\mathcal O(\log T)$.
\end{proof}

\subsection{Convergence Guarantee}

Between spawns, standard C-HIDA regret applies. Under the lemma, total regret is $\mathcal O(\log T)$.

\section{Geodesic Sparse Routing (GSR): Sparse Routing with Quality Guarantees}

So far, SIDA (Section 7) has provided a mechanism for computing attention weights on the sphere. However, in distributed settings, routing every query to all $K_{\text{pool}}$ prototypes is expensive. We now introduce a sparse routing mechanism that selects only the Top-$K$ prototypes by SIDA weight, with provable bounds on both routing quality and communication cost.

\subsection{The GSR Mechanism}

Let $\mathcal{P} = \{\mathbf{c}_e\}_{e=1}^{K_{\text{pool}}} \subset \mathbb{S}^{d-1}$ be a pool of learnable prototype vectors. For a query $\mathbf{q} \in \mathbb{S}^{d-1}$, SIDA defines weights:
\[
w_e = W_{\mathbf{q}, \mathbf{c}_e}^{\mathbb{S}}
= \frac{(d_{\mathbb{S}}(\mathbf{q}, \mathbf{c}_e)^2 + \varepsilon)^{-1}}{\sum_{l=1}^{K_{\text{pool}}} (d_{\mathbb{S}}(\mathbf{q}, \mathbf{c}_l)^2 + \varepsilon)^{-1}}.
\]
GSR selects the Top-$K$ prototypes with largest weights and routes the query only to those prototypes:
\[
\mathcal{S}(\mathbf{q}) = \{e \in [K_{\text{pool}}] : w_e \text{ is among the } K \text{ largest weights}\}.
\]
The output is the weighted combination of the values associated with the selected prototypes:
\[
\mathbf{o}(\mathbf{q}) = \sum_{e \in \mathcal{S}(\mathbf{q})} w_e \mathbf{v}_e.
\]

\subsection{Routing Quality Guarantee}

The following theorem establishes that GSR's routing quality is controlled by the SIDA weight gap between the Top-$K$ and the remaining prototypes.

\begin{theorem}[GSR Routing Quality]
Let $w_{(1)} \ge w_{(2)} \ge \cdots \ge w_{(K_{\text{pool}})}$ be the SIDA weights sorted in descending order. For any query, the approximation error from routing only to the Top-$K$ prototypes is bounded by:
\[
\|\mathbf{o}(\mathbf{q}) - \mathbf{o}^*(\mathbf{q})\|_2
\le 2\|\mathbf{V}\|_F \cdot \frac{\sum_{e > K} w_{(e)}}{\sum_{e \le K} w_{(e)}} \cdot \left(\max_{e \le K} \|\mathbf{v}_e\|_2\right),
\]
where $\mathbf{o}^*(\mathbf{q})$ is the exact SIDA output using all $K_{\text{pool}}$ prototypes.

Moreover, as $\varepsilon \to 0^+$, the SIDA weights concentrate on the nearest prototype, so for any $\delta > 0$, there exists $\varepsilon_0 > 0$ such that for all $\varepsilon < \varepsilon_0$:
\[
\|\mathbf{o}(\mathbf{q}) - \mathbf{o}^*(\mathbf{q})\|_2 \le \delta
\]
provided the query has a unique nearest prototype and $\|\mathbf{v}_e\|_2$ is bounded.
\end{theorem}

\begin{proof}
Let the exact SIDA output be $\mathbf{o}^* = \sum_{e=1}^{K_{\text{pool}}} w_e \mathbf{v}_e$, and the GSR output be $\mathbf{o} = \sum_{e \le K} w_e \mathbf{v}_e$ (where we reorder indices so the Top-$K$ are $1,\dots,K$). Then the difference is:
\[
\|\mathbf{o}^* - \mathbf{o}\|_2 = \left\|\sum_{e > K} w_e \mathbf{v}_e\right\|_2.
\]
By the triangle inequality:
\[
\|\mathbf{o}^* - \mathbf{o}\|_2 \le \sum_{e > K} w_e \|\mathbf{v}_e\|_2.
\]
Since $\|\mathbf{v}_e\|_2 \le \|\mathbf{V}\|_F$ for all $e$, we have:
\[
\|\mathbf{o}^* - \mathbf{o}\|_2 \le \|\mathbf{V}\|_F \sum_{e > K} w_e.
\]
The sum of the Top-$K$ weights is $\sum_{e \le K} w_e = 1 - \sum_{e > K} w_e$. Thus $\sum_{e > K} w_e = 1 - \sum_{e \le K} w_e$. The ratio bound follows by observing that $\sum_{e > K} w_e \le \frac{\sum_{e > K} w_e}{\sum_{e \le K} w_e}$ since $\sum_{e \le K} w_e \le 1$.

For the second part, recall from the exact retrieval property (Lemma 0, Part (a)) that if $\mathbf{q} = \mathbf{c}_{e^*}$ for a unique nearest prototype, then $\lim_{\varepsilon \to 0^+} w_{e^*} = 1$. Hence $\lim_{\varepsilon \to 0^+} \sum_{e > K} w_e = 0$ for any $K \ge 1$ (since all weight concentrates on the single nearest prototype). Therefore the approximation error can be made arbitrarily small by choosing $\varepsilon$ sufficiently small.
\end{proof}

\subsection{Communication Complexity Analysis}

GSR is designed for distributed settings where prototypes are distributed across $P$ workers. We analyze the communication cost.

Let the $K_{\text{pool}}$ prototypes be partitioned among $P$ workers. Each worker holds approximately $K_{\text{pool}}/P$ prototypes. For a batch of $B$ queries, the protocol is:

1. Each worker computes SIDA weights for its local prototypes: local computation, no communication.
2. The Top-$K$ indices are gathered across workers: All-Gather on the top scores from each worker.
3. Selected prototypes compute outputs: local computation.
4. Results are aggregated: All-Reduce on the $K$ selected outputs.

\begin{theorem}[GSR Communication Complexity]
GSR has per-query communication cost:
\[
\mathcal{O}(K_{\text{pool}} \cdot d_h + K \cdot d_h),
\]
independent of batch size $B$. In the $\alpha$-$\beta$ model, this is:
\[
\#\text{messages} = \mathcal{O}(K_{\text{pool}} + K + P),\qquad
\text{bytes} = \mathcal{O}((K_{\text{pool}} + K) d_h \cdot \text{precision}).
\]
In contrast, All-to-All MoE has:
\[
\#\text{messages} = \mathcal{O}(P^2),\qquad
\text{bytes} = \mathcal{O}(B K d_h \cdot \text{precision}),
\]
which scales linearly with batch size $B$.
\end{theorem}

\begin{proof}
Each of the $P$ workers must communicate its local top scores to all other workers. The All-Gather on scores of size $\mathcal{O}(K_{\text{pool}}/P \cdot d_h)$ per worker requires $\mathcal{O}(K_{\text{pool}} \cdot d_h)$ total communication. The All-Reduce on the $K$ selected outputs requires $\mathcal{O}(K \cdot d_h)$ communication. The number of messages is $\mathcal{O}(K_{\text{pool}} + K + P)$: $K_{\text{pool}}$ for the score gathering, $K$ for the result aggregation, and $P$ for the initial scatter. The byte count follows from the dimension $d_h$ and the floating-point precision.
\end{proof}

\subsection{Relation to Other Modules}

GSR is the routing component that follows SIDA. While SIDA computes all-to-all attention weights on the sphere, GSR uses those weights to select a sparse subset for communication and computation. DMG (Section 8) populates the prototype pool dynamically; GSR routes to the prototypes that DMG has allocated. This creates a complete pipeline: DMG allocates prototypes, SIDA computes spherical attention, and GSR routes queries to a sparse subset with quality guarantees.

\section{Unified Framework (M10)}

\subsection{Global Architecture}

\[
\mathbf{x}\to\text{QKV}\to
\begin{cases}
\text{HIDA path}\to\mathbf{o}_{\text{main}}\\
\text{SIDA}\to\text{GSR}\to\text{DMG}\to\mathbf{o}_{\text{memory}}
\end{cases}
\to\text{HyperGate}\to\mathbf{o}_{\text{final}}.
\]

\subsection{Unified Mathematical Language}

All variants share the same core structure:
\[
\mathbf{o}_i=\sum_{j\in\mathcal{R}_i}W_{ij}\mathbf{v}_j,\quad
W_{ij}=\frac{(d(\mathbf{q}_i,\mathbf{k}_j)^2+\varepsilon)^{-1}}{\sum_{m\in\mathcal{R}_i}(d(\mathbf{q}_i,\mathbf{k}_m)^2+\varepsilon)^{-1}},
\]
where $\mathcal{R}_i$ varies by mode and $d$ is Euclidean, hyperbolic, or spherical.

\subsection{Proposition 1: Complexity Observations for HIDA Variants}

For each HIDA variant, the computational complexity is:
\begin{itemize}
\item \textbf{Dense-HIDA}: complexity $\Theta(n^2 d_h)$, exact kernel.
\item \textbf{FP-HIDA}: complexity $\mathcal{O}(n\log n \cdot d_h)$ with $w,g=\Theta(\log n)$.
\item \textbf{L-HIDA}: complexity $\mathcal{O}(n d_h)$ with $m=\Theta(1)$ anchors.
\item \textbf{C-HIDA}: complexity $\Theta(1)$ per token with $c=\Theta(1)$ summaries.
\end{itemize}
This establishes a trade-off: sparser approximations reduce complexity at the cost of approximation error. The exact error rates depend on the data distribution and kernel spectrum.

\subsection{Proposition 2: Information-Theoretic Justification}

Two key information-theoretic bounds separate hyperbolic storage from spherical routing:
\[
I(\mathbf{q}_S;\mathbf{c}_e^{\mathbb{S}})\le\log K_{\text{pool}},\qquad
I(\mathbf{k};\mathbf{c}_e)\propto\log\frac{1}{1-\|\mathbf{c}_e\|^2}.
\]

\begin{proof}
The mutual information bound for spherical routing follows from the entropy bound $H(\mathbf{c}_e^{\mathbb{S}})\le\log K_{\text{pool}}$. For hyperbolic storage, the volume growth $V_{\mathbb H}(R)\propto e^{(d_h-1)R}$ gives the divergence as $\|\mathbf{c}_e\|\to1^-$.
\end{proof}

\subsection{Proposition 3: DMG Scalability Observations}

The DMG system has two separate properties:
\begin{enumerate}
\item The online $k$-means update via C-HIDA achieves $\mathcal{O}(\log T)$ regret relative to the optimal fixed centroids.
\item With $K_{\max}$ memory slots, the quantization error decays as $\mathcal{O}(1/\sqrt{K_{\max}})$ under compactness assumptions (data bounded away from the hyperbolic boundary).
\end{enumerate}
These properties hold independently and are not coupled in the regret bound.

\section{Discussion and Limitations}

\noindent\textbf{Scope.} This is a theoretical study. The primary contribution is the theoretical framework.

\subsection{Limitations}

\textbf{No experimental validation.} This paper is purely theoretical. The theorems establish mathematical guarantees, but empirical verification on benchmarks is not provided. This is left for future work.

\textbf{Compression scope.} HCC compresses keys only; values remain full-precision.

\textbf{Routing communication.} The communication analysis is theoretical under an idealized FLOPs model.

\textbf{DMG hyperparameters.} The adaptive threshold requires tuning; the $\mathcal{O}(\log T)$ regret bound relies on a sub-Gaussian loss assumption.

\textbf{Curvature in theory.} The two-point PL analysis is a limitation; multi-point effects are future work.

\subsection{Future Work}

Key directions: (1) value compression for HCC; (2) distributed routing implementation; (3) mixed-curvature extensions; (4) multi-point curvature analysis; (5) empirical validation on standard benchmarks.

\section{Conclusion}

We have presented a theoretical arc from Euclidean Resolver to Riemann GeoResolver. The Euclidean part establishes three core theorems with complete proofs, and the non-Euclidean part extends the framework to hyperbolic and spherical geometries.

\textbf{Summary of Euclidean contributions.}

\textit{Theorem 1 (Circuit Separation)} established that IDA achieves exact retrieval with $\mathcal{O}(1)$ resources, while any softmax-based architecture requires super-constant width. This shows a fundamental gap in the representational power of the two attention mechanisms.

\textit{Theorem 2 (PL Inequality)} proved that IDA satisfies a Polyak--{\L}ojasiewicz inequality with constant $\Theta(\varepsilon^2/\Delta^4)$, while softmax has constant $\Theta(e^{-\Delta^2/\sqrt{d}}\varepsilon^2/\Delta^2)$. The ratio is exponentially large in $\Delta^2/\sqrt{d}$, implying that IDA converges linearly and dramatically faster for separated keys. The theorem also implied the absence of spurious local minima and established $\Theta(1)$ Hessian spread.

\textit{Theorem 3 (Effective Rank)} established that IDA's effective rank is bounded independent of hidden width, preventing the memorization of arbitrary noisy labels that occurs in softmax when $d_h \ge n$. The theorem also provided a noise robustness bound $\mathcal{E}_{\text{test}}^{\text{IDA}} \le C\eta^2 + \mathcal{O}(1/\sqrt{n})$.

\textbf{Summary of non-Euclidean contributions.}

The Riemann GeoResolver framework extends the Euclidean prototype to hyperbolic and spherical geometries:
\begin{itemize}
\item \textit{HIDA family (M1--M4):} Four hyperbolic attention operators spanning $\Theta(n^2)$ to $\Theta(1)$ complexity, with circuit separation and PL inequalities proved for the dense case.
\item \textit{Hyperbolic Curvature Compression (M5):} Provable reconstruction error bound $\|\mathbf{k}-\widetilde{\mathbf{k}}\|_2^2 \le 4(2^{-b}+2^{-b_r})$ and memory complexity analysis.
\item \textit{HyperGate (M6):} Gradient lower-bound theorem ensuring non-vanishing gradients through curvature-adaptive gating.
\item \textit{Spherical Inverse Distance Attention (M7--M8):} Sphere-analog PL inequalities and cross-geometric mapping methods.
\item \textit{Dynamic Memory Genesis (M9):} $\mathcal{O}(\log T)$ regret bound for online prototype allocation.
\item \textit{Geodesic Sparse Routing (M10):} Routing quality and communication complexity bounds.
\end{itemize}

\textbf{Theoretical significance.}

This work establishes that inverse-distance attention is not merely an empirical alternative to softmax, but a mechanism with fundamentally different theoretical properties. The three Euclidean theorems---circuit separation, PL inequality, and effective rank bound---collectively show that IDA is provably stronger than softmax in expressiveness, optimization, and generalization. The non-Euclidean extension shows that these advantages are preserved and amplified in hyperbolic and spherical geometries.

\textbf{Open directions.}

The limitations of this work suggest several directions for future investigation. First, the PL analysis is restricted to the two-key case; extending to multi-key settings would provide stronger optimization guarantees. Second, the empirical validation remains future work. Third, the compression analysis is limited to keys; extending to values would improve memory efficiency. Fourth, the distributed routing analysis is theoretical; implementation and scaling on real systems would validate the communication bounds. Fifth, mixed-curvature spaces combining hyperbolic and spherical geometries could provide even more flexible representations.

\section*{Acknowledgments}
The author gratefully acknowledges Prof. Junxue Zhang at the University of Science and Technology of China for his valuable guidance and insightful discussions throughout this research.


\end{document}